\documentclass{amsart}

\usepackage[english]{babel}
\usepackage{amsmath}
\usepackage{amssymb}
\usepackage{amsthm}
\usepackage{latexsym}
\usepackage{mathtools}
\usepackage{thmtools}
\usepackage{amsfonts}
\usepackage[T1]{fontenc}
\usepackage{enumerate}
\usepackage{parskip}

\usepackage[top=3cm,bottom=3cm,left=3.5cm,right=3.5cm]{geometry}

\newtheorem{theorem}{Theorem}
\newtheorem{corollary}[theorem]{Corollary}
\newtheorem{proposition}[theorem]{Proposition}
\newtheorem{lemma}[theorem]{Lemma}
\newtheorem{question}[theorem]{Question}

\theoremstyle{definition}
\newtheorem{definition}[theorem]{Definition}

\theoremstyle{remark}
\newtheorem{remark}[theorem]{Remark}

\newcommand{\rank}{\mathrm{rank}}
\newcommand{\vvirg}{ , \dots , }

\newcommand{\textsum}{{\textstyle \sum}}
\newcommand{\textprod}{{\textstyle \prod}}

\newcommand{\textbinom}[2]{{\textstyle \binom{#1}{#2}}}
\newcommand{\textfrac}[2]{{\textstyle \frac{#1}{#2}}}

\newcommand{\bfR}{\mathbf{R}}

\newcommand{\bfa}{\mathbf{a}}
\newcommand{\bfb}{\mathbf{b}}

\newcommand{\bfg}{\mathbf{g}}
\newcommand{\bfh}{\mathbf{h}}

\newcommand{\bfp}{\mathbf{p}}

\newcommand{\bfx}{\mathbf{x}}
\newcommand{\bfy}{\mathbf{y}}

\newcommand{\calJ}{\mathcal{J}}

\newcommand{\bbC}{\mathbb{C}}

\newcommand{\bbN}{\mathbb{N}}

\newcommand{\bbP}{\mathbb{P}}

\newcommand{\bbR}{\mathbb{R}}

\newcommand{\bbZ}{\mathbb{Z}}

\newcommand{\frakS}{\mathfrak{S}}

\renewcommand{\phi}{\varphi}

\renewcommand{\theta}{\vartheta}

\renewcommand{\tilde}[1]{\widetilde{#1}}

\newcommand{\id}{\mathrm{id}}
\newcommand{\Id}{\mathrm{Id}}

\DeclareMathOperator{\End}{End}
\DeclareMathOperator{\Sym}{Sym}

\newcommand{\bfVNP}{\mathbf{VNP}}
\newcommand{\bfVP}{\mathbf{VP}}

\newcommand{\perm}{\mathrm{perm}}
\newcommand{\Pow}{\mathrm{Pow}}

\newcommand{\trace}{\mathrm{trace}}

\renewcommand{\det}{\mathrm{det}}

\newcommand{\uline}[1]{\underline{#1}}

\newcommand{\rk}{\mathrm{rank}}
\newcommand{\BC}{\mathbb C}

\newcommand{\La}[1]{\Lambda^{#1}}

\newcommand{\tim}{\mathrm{image}}
\newcommand{\IMM}{\mathrm{IMM}}
\newcommand{\ur}{\underline {\bfR}} 

\newcommand{\xto}[1]{\xrightarrow{#1}}

\newcommand{\Mat}{\mathrm{Mat}}

\title[Explicit polynomials with maximal spaces of partials]{Explicit polynomial  sequences  with maximal spaces of partial derivatives and a question of K. Mulmuley}

\author[F. Gesmundo]{Fulvio Gesmundo}
\address[F. Gesmundo]{QMATH, Dept. of Mathematical Sciences, Univ. of Copenhagen, Denmark}
\email{fulges@math.ku.dk}
\thanks{Gesmundo is supported by the European Research Council (ERC Grant Agreement no. 337603) and VILLUM FONDEN via the QMATH Centre of Excellence (Grant no. 10059)}

\author[J. M. Landsberg]{Joseph M. Landsberg}
\address[J. M. Landsberg]{Dept. of Mathematics, Texas A\&M University, College Station, TX}
\email{jml@math.tamu.edu}
\thanks{Landsberg is supported by NSF grants DMS-1405348 and CCF-1814254.}


\begin{document}

\begin{abstract}
We answer a question of K. Mulmuley: In \cite{2016arXiv160902103E} it was shown that the method of shifted partial derivatives cannot be used to separate the padded permanent from the determinant. Mulmuley asked if this \lq\lq no-go\rq\rq\ result could be extended to a model without padding. We prove this is indeed the case using the iterated matrix multiplication polynomial. We also provide several examples of polynomials with maximal space of partial derivatives, including the complete symmetric polynomials. We apply Koszul flattenings to these polynomials to have the first explicit sequence of polynomials with symmetric border rank lower bounds higher than the bounds attainable via partial derivatives.
\end{abstract}

\maketitle

\section{Introduction}
 
 Let $S^d\BC^N$ denote the space of homogeneous polynomials of degree $d$ in $N$ variables and let $p\in S^d\BC^N$. Let $S^e\BC^{N*}$ denote the space of homogeneous differential operators of order $e$ with constant coefficients, which acts on $S^d\BC^N$ when $e\leq d$. The {\it $e$-th partial derivative map}  of $p$  (or \emph{$e$-th flattening of $p$}) is
\begin{align} \label{flateqn}
p_{e,d-e}: S^e\BC^{N*}& \to S^{d-e}\BC^N\\
\nonumber D & \mapsto D(p).
\end{align}
 We call the image of $p_{e,d-e}$ the $e$-th \emph{space of partial derivatives} of $p$; it is straightforward to verify that $\rank (p_{e,d-e}) = \rank (p_{d-e,e})$ and that given $e' \leq e \leq d/2$, if $p_{e,d-e}$ has full rank then $p_{e',d-e'}$ has full rank.

Let $M\geq N$. Choose a linear inclusion $\bbC^N \subseteq \bbC^M$, so that polynomials in $N$ variables can be regarded as polynomials in $M$ variables which happen to use only $N$ of them. A polynomial $p\in S^d\BC^N$ is a {\it degeneration} of $q\in S^d\BC^M$, if  $p\in \overline{GL_M\cdot q}\subseteq S^d\BC^M$, where the overline denotes closure, equivalently in the usual (Euclidean) topology or in the Zariski topology.  Similarly, $p$ is a \emph{specialization} of $q$ if $p \in \End_{M} \cdot q \subseteq S^d \bbC^M$. Notice that if $p$ is a specialization of $q$ then it is a degeneration of $q$. In complexity theory, one is interested in finding obstructions to specialization of a polynomial $p$ to a polynomial $q$.  

A common strategy to determine obstructions to degeneration (hence, to specialization) of a polynomial $q$ to a polynomial $p$ is to find closed conditions that $q$ and every element of $GL_M \cdot q$ satisfy, but $p$ does not satisfy. Typical examples of this method are flattening techniques (see, e.g., \cite{MR2901512}), which exploit the semi-continuity of matrix rank, namely that $\{ X \in \Mat_{\ell \times \ell} : \rank( X) \leq r\}$ is a closed subset of $\Mat_{\ell \times \ell}$ (both in the usual and in the Zariski topology).

 The {\it method of partial derivatives} is one such example: since the entries of $q_{e,d-e}$ (as a matrix expressing the map $S^e \bbC^{N*}  \to S^{d-e} \bbC^N$) are continuous in the coefficients of $q$, whenever $p$ is a degeneration of $q$, semi-continuity of matrix rank guarantees that $\rank(p_{e,d-e})\leq \rank(q_{e,d-e})$ for all $e$. Therefore, comparing the ranks of the partial derivatives maps of $p$ and $q$ for any $e$, one can prove that $p$ is not a degeneration of $q$ (and thus nor is $p$ a specialization of $q$). 

The method of partial derivatives  dates back to Sylvester in 1852 \cite{sylvestercat}, who called the maps \eqref{flateqn} {\it catalecticants}. These maps have been used to obtain lower bounds on the Waring rank, Waring border rank and cactus border rank  of polynomials (see, e.g., \cite{MR3121848,MR1735271}). The {\it symmetric} or {\it Waring rank} of a polynomial $p\in S^d\BC^N$ is the smallest $r$ such that $p=\sum_{j=1}^r \ell_j^d$ where $\ell_j\in \BC^N$ are linear forms. One writes $\bold R_S(p)=r$. The {\it symmetric } or {\it Waring border rank} of $p$ is the smallest $r$ such that $p$ is a limit of polynomials of Waring rank $r$, and one writes  $\ur_S(p)=r$.  The ranks of the partial derivatives maps give lower bounds for the symmetric border rank of $p$: $\ur_S(p) \geq \max_{e} \{ \rank(p_{e,d-e})\}$. From a complexity theory perspective, Waring rank captures the complexity of a polynomial in the model of depth three powering circuits, or $\Sigma\Lambda\Sigma$ circuits, namely depth three arithmetic circuit whose first and third 
layers consist of addition gates and whose middle layer consists of powering gates, sending $z\mapsto z^d$ for some $d$.

 In \cite{MR0480472,MR515043}, it was shown that for a general polynomial $p$  all the maps $p_{e,d-e}$ are of maximal rank. When the second author was preparing \cite{Lsimons}, he asked several experts if they knew of an explicit  sequence of polynomials (e.g., in the complexity class $\bfVNP$) with partial derivatives of maximal rank, as the standard references \cite{MR1735271} in mathematics and  \cite{MR2901512} in computer science did not have one. Those asked did not furnish any example, so we wrote down several, see below. One example we found surprised us: the polynomial $(x_1^2+\cdots +x_n^2)^k$, because it is in  the complexity class  $\bfVP_e$ of sequences of polynomials admitting polynomial size formulas. It turns out this example had been discovered by Reznick in 1991 \cite[Thm. 8.15]{MR1096187}, and in the same memoir he describes an explicit sequence that essentially dates back to  Bierman \cite{MR1547094} (the proof, if not the statement appeared in 1903), see below.

 Let $p_{n,d}=x_1^d +\cdots + x_n^d$ denote the power sum polynomial  of degree $d$ in $n$ variables and $h_{n,d}=\sum_{\vert \alpha \vert=d}x_1^{\alpha_1}\cdots x_{n}^{\alpha_n}$ the complete symmetric polynomial  of degree $d$ in $n$ variables. 

For the following polynomial sequences, all partial derivatives map have full rank:
\begin{itemize}
 \item $P_{Bier, n,d}:=\sum_{\vert \alpha\vert = d}(\alpha_{1}x_1+\cdots + \alpha_{n} x_{n})^d$ where $\alpha$ ranges over all  multi-indices  $(\alpha_{1}\vvirg  \alpha_{n})$  of non-negative integers  such that $\alpha_{1}+\cdots + \alpha_{n}=d$ i.e., exponents of monomials of degree $d$ in $n$ variables (Bierman \cite{MR1547094},Reznick \cite{MR1096187});
 
 \item $f_{n,k}:=(p_{n,2})^k\in S^{2k}\BC^n$, (Reznick \cite{MR1096187}, a proof is given in \S\ref{easycataproofs});
 
 \item $\tilde f_{n,k}:=p_{n,1}f_{n,k}\in S^{2k+1}\BC^n$ (Theorem \ref{thm: flattenings of ellqk hit everything});
 
 \item $h_{n,d}\in S^d\BC^n$ (Theorem \ref{thm: hnd have maximal catalecticant rank}).
 \end{itemize}
 
 When $n,d$ are polynomially related, $h_{n,d}\in \bfVP_s$, the complexity class determined by the determinant,  because the complete symmetric functions can be expressed as a determinant of a matrix whose entries are power sum functions. On the other hand, the polynomials $f_{n,k}$ and $\tilde{f}_{n,k}$ belong to the complexity class $\bfVP_e$ of polynomials admitting a polynomial size formula. The fact that there are elements in $\bfVP_e$ having partial derivatives map of full rank suggests that the $\Sigma\Lambda\Sigma$ model is quite weak.

The {\it method of shifted partial derivatives} is a variant of the method of partial derivatives. It was introduced in \cite{DBLP:journals/eccc/Kayal12} and exploited in \cite{gupta4} to prove super-polynomial complexity lower bounds for depth four circuits for the permanent (and determinant). In the same paper the authors ask if the method could be used to approach Valiant's conjecture, that is to separate the class $\bfVP$ from the class $\bfVNP$.

For $p\in S^d\BC^N$  the method of shifted partials is based on the study of the following maps (for judiciously chosen $e$ and $\tau$):
\begin{align*}
p_{(e,d-e)[\tau]}:S^{e}\BC^{N^*}\otimes S^\tau \BC^N& \to S^{d-e+\tau}\BC^N\\
D\otimes q & \mapsto qD(p).
\end{align*}
 Notice that if $\tau = 0$, then $p_{(e,d-e)[\tau]}$ is the partial derivative map defined in \eqref{flateqn}. Let  $\langle \partial^{=e}p\rangle_{=\tau}:=p_{(e,d-e)[\tau]}(S^{e}\BC^{N^*}\otimes S^\tau \BC^N)$.

Again, semi-continuity of matrix rank guarantees that if $p$ is a degeneration of $q$, then $\dim \langle \partial^{=e} q \rangle_{=\tau} \geq  \dim \langle \partial^{=e}p \rangle_{=\tau}$ for all $\tau$ and the method of shifted partials can be used to prove that $p$ is not a degeneration of $q$ by showing that $\dim \langle \partial^{=e}p \rangle_{=\tau} > \dim \langle \partial^{=e} q \rangle_{=\tau}$ for some $\tau$.

 There is a geometric interpretation of the image of the shifted partial derivative map: given $p \in S^d \bbC^n$, let $V(p) \subseteq \bbP (\bbC^n)^*$ be the hypersurface of degree $d$ cut out by $p$. The image of $p_{e,d-e}$ generates an ideal in $\Sym(\bbC^n)$ that we denote by $\calJ_e(p)$; it cuts out a subvariety of $V(p)$ that is called the $e$-th \emph{Jacobian locus of} $p$. The image of the shifted partials map $p_{(e,d-e)[\tau]}$ is the component of degree $d+\tau$ of $\calJ_e(p)$, in other words, the value of the Hilbert function of $\calJ_e(p)$ in degree $d+\tau$.  In particular, the study of the ranks of the shifted partials maps of $p$ is equivalent to the study of the growth of the ideals $\calJ_e(p)$ and more precisely of the growth of their Hilbert functions.  

 The results of \cite{2016arXiv160902103E} show that the method of shifted partial cannot be used to separate $\bfVP_s$ from $\bfVNP$ in the classical formulation of Valiant's conjecture, where one seeks for a super-polynomial lower bound on the so-called determinantal complexity of the \emph{padded} permanent polynomial, namely on the smallest possible $n(m)$ such that $z^{n(m)- m} \perm_m$ is a specialization of $\det_{n(m)}$. Informally, the proof of \cite{2016arXiv160902103E} exploits the padding $z^{n-m}$ to prove that the shifted partial spaces of $z^{n-m} \perm_m$ do not grow fast enough to give a super-polynomial separation from the shifted partial spaces of $\det_n$.
 
This motivates the question on whether the method of shifted partials can be used to achieve a super-polynomial lower bound in a model which does not require padding.

\begin{definition}
Given $n,d$, let $X_\alpha = ((\xi_\alpha)^i_j)_{i,j = 1 \vvirg n}$ be $n \times n$ matrices of indeterminates for $\alpha = 1 \vvirg d$. The \emph{$(n,d)$-iterated matrix multiplication} polynomial $\IMM^d_n\in S^d(\BC^{dn^2})$ is 
\[
\IMM^d_n: (X_1\vvirg X_d)\mapsto \trace(X_1\cdots X_d).
\] 

Let $(\xi^i_j)_{i,j = 1 \vvirg n}$ be an $n\times n$ matrix of indeterminates. The $(n,d)$-\emph{matrix powering} polynomial $\Pow^d_n\in S^d\BC^{n^2}$ is
\[
 \Pow^d_n : X \mapsto \trace(X^d)
\]
\end{definition}

By \cite{Nisan:1991:LBN:103418.103462}, the polynomials $\IMM^d_n$ and $\Pow^d_n$ can be used to define $\bfVP_s$-complete sequences without the use of padding. More precisely, a sequence of homogeneous polynomials $\{ f_m \}_{m \in \bbN}$ with $f_m \in S^{d_m} \bbC^{M_m}$ (with $d_m, M_m$ growing polynomially in $m$) is in $\bfVP_s$ if and only if there exists a function $n(m)$ (resp. $n'(m)$) growing polynomially in $m$, such that $f_m$ is a specialization of $\IMM^{d_m}_{n(m)}$ (resp. $\Pow^{d_m}_{n'(m)}$). Remarkably, in this case, the model does not require padding and allows one to compare $\IMM^d_n$ and $\Pow^d_n$ directly with the sequence of polynomials.

In particular, Valiant's $\bfVP_s \neq \bfVNP$ conjecture can be rephrased by stating that there is no polynomially bounded function $n(m)$ such that the permanent polynomial $\perm_m$ is a specialization of $\IMM^m_{n(m)}$ (or of $\Pow^m_{n(m)}$). Note that $\Pow^m_n$ is a specialization of $\IMM^m_n$. 

We prove that the method of shifted partials cannot be used to achieve a super-polynomial separation between $\perm_m$ and $\IMM^m_n$: 

\begin{theorem}\label{noshiftthm}
If $n> m^5$, then $\perm_m$ cannot be separated from $\IMM^m_{n}$ by the method of shifted partial derivatives. More precisely, given any linear inclusion $\BC^{m^2}\subseteq \BC^{ {m} n^2}$ considering $\perm_m\in S^m\BC^{ {m} n^2}$ as a polynomial that just involves $m^2$ of the $m n^2$ variables, then for all choices of $e$,$\tau$, $\dim  \langle \partial^{=e}(\perm_m\in S^m\BC^{ {m}n^2})\rangle_{=\tau}\leq \dim  \langle \partial^{=e}\IMM^m_n\rangle_{=\tau}$.
\end{theorem}

 \subsection*{Additional results} We give {\it a priori} upper bounds for the utility of Koszul flattenings, another variant of the partial derivatives map, in comparing the complexity of polynomials (Proposition \ref{kozapriori}). We show that these bounds are sharp for the first Koszul flattenings in low dimensions and degree (Remark \ref{koscomputer}). We  obtain explicit (but not sharp) lower bounds for the Koszul flattenings of $\tilde f_{n,k}$, showing that one obtains better Waring border rank lower bounds with this method than by the method of partial derivatives (Proposition \ref{fknkosz}). Ironically, now the simple polynomial $\tilde f_{n,k}$ has the highest Waring  border rank lower bound of all explicit polynomials of odd degree.

 \subsection*{Related work} Let $e_{n,d}=\sum_{1\leq i_1<i_2< \cdots <i_d\leq n} x_{i_1}\cdots x_{i_d}\in S^d\BC^n$ denote the elementary symmetric polynomial of degree $d$ in $n$ variables. The complexity of $e_{n,d}$ has been well studied: its symmetric border rank is bounded below by $\binom{n}{\lfloor \frac {d}{2}\rfloor}$ because this is the rank of $(e_{n,d})_{ \lfloor d/2 \rfloor, \lceil d/2 \rceil}$ (see e.g. \cite{NisWig:LowerBoundsArithCircPartDer}); when $d$ is odd, its symmetric rank is $\sum_{i=0}^{\lfloor \frac d2\rfloor}\binom{n}{i}$  when $d$ is odd and there is a similar formula for even $d$ \cite{MR3440150}. Its padded version can be computed  by a homogeneous depth three ($\Sigma\Pi\Sigma$)  circuit of size $n^2$ (due to Ben-Or), and when $\log(n)\leq d\leq \frac {2n}3$,  one has the lower bound of $\max\{ \Omega (\frac{n^2}{d}),\Omega(n d)\}$ from \cite{MR1867306} for its depth three circuit size. The lower bounds  appear to translate to complete symmetric functions; however the upper bound relies on the generating function for the elementary symmetric functions being a product of linear forms, whereas the complete symmetric functions have generating function $\prod_{i=1}^n (1-x_it)^{-1}$. The gap between the padded and unpadded depth three circuit complexity may have led researchers to
 think the results of \cite{2016arXiv160902103E} might fail in a model without padding. For this reason, K. Mulmuley suggested that we investigate the barriers of the method of shifted partials in the unpadded setting (although Mulmuley himself anticipated our answer). Similar concerns were shown after the results of \cite{IkPa:Rectangular_Kron_in_GCT} and \cite{BuIkPa:no_occurrence_obstructions_in_GCT} on the Geometric Complexity Theory program, and were partially addressed in \cite{GesIkPa:GCT_and_matrix_pow}, exploiting the same homogenization result of \cite{Nisan:1991:LBN:103418.103462} that we use in this work.
 
 The shifted partial derivative complexity of elementary symmetric polynomials is studied in \cite{MR3419496}, where strong lower bounds are proved, which in turn give complexity lower bounds for depth four circuits.
 
\subsection*{Acknowledgements} We thank C. Yan for discussions on the Gessel-Viennot method, B. Reznick and Z. Teitler for historical information, and N. Stein and Mathoverflow for the proof of Lemma \ref{lemma: principal minor on indep rows}. We thank the Santa Fe Institute and the organizers of the working group in \emph{Geometric Complexity Theory} in December 2016 that inspired this work. We also thank the anonymous referees for numerous suggestions to improve the exposition and on the estimate in Case 2 of the proof of Theorem \ref{noshiftthm}.

\section{Proofs that $f_{n,k}$, $\tilde f_{n,k}$ have maximal partial derivatives}\label{easycataproofs}
Introduce the notation $q_n :=p_{n,2} =x_1^2+\cdots + x_n^2$, $\ell_n:= p_{n,1} = x_1 + \cdots + x_n$.

In this section, we prove that $f_{n,k}=(p_{n,2})^k$ and $\tilde{f}_{n,k}=p_{n,1}f_{n,k}$ have partial derivatives maps of maximal rank, or equivalently that their spaces of partial derivatives have the maximal possible dimension. Although the result for $f_{n,k}$ already appeared in \cite[Thm. 8.15]{MR1096187}, we include our proof in this section because it seems less involved and more accessible than the original one; moreover the proof of the result for $\tilde{f}_{n,k}$ relies on the method that we use to prove the result for $f_{n,k}$ (see Remark \ref{rmk: Lemma holds for every quadric}).

 \begin{proposition}\label{prop: hit everything with deg one differentials}
The space of first  derivatives of homogeneous  degree $d+2$ polynomials of the form $h q_n$, where $h$ runs over the space of  homogeneous  degree $d$ polynomials,  equals the space of all homogeneous polynomials of degree $d+1$. In symbols, setting $L = \{ \partial ( h q_n  ) : \partial \in (\bbC^n)^*, h \in S^{d} \bbC^n\}$, we have $L = S^{d+1} \bbC^n$. 
\end{proposition}

\begin{proof}
If $d = 0$, the statement holds as 
\[
\frac{\partial}{\partial x_i} q_n = 2x_i.
 \]

Let $d \geq 1$. Consider a monomial $\bfx ^\alpha$ of degree $d-1$, where $\alpha$ is a multi-index, $\vert \alpha \vert = d-1$. For   $i\in \{ 1\vvirg n\}$,  define 
\[
g_i := \frac{\partial}{\partial x_i}(q_n x_i \bfx^\alpha  ) = 2x_i^2 \bfx^\alpha    + (\alpha_i+1)q_n \bfx^\alpha  = \left[ 2x_i^2 + (\alpha_i+1)q_n \right] \bfx^\alpha.
\]

For every $i$, $g_i \in L$. Let $\bfg = (g_1 \vvirg g_n)^T$ and $\bfp = (x_1^2 \vvirg x_n^2)^T$ be column vectors. We have
\[
 \bfg =  \left[ (2\Id + A) \bfp \right] \bfx^\alpha
\]
where $A$ is the $(n+1) \times (n+1)$ rank $1$ matrix whose entries in the $i$-th row are $\alpha_i + 1$ for $i=1 \vvirg n$. Let $\mathbf{1} = (1 \vvirg 1 )^T$ and $\bfa = (\alpha_1+1 \vvirg \alpha_n+1)^T$ and notice $A = \mathbf{1} \bfa ^T$. In particular by the Sylvester determinant identity (see, e.g., \cite[Chap. 1, Prob. 3.1]{MR1277174}) $\det(2\Id+A) \neq 0$, so $2 \Id + A$ is invertible.

Therefore, $x_i^2 \bfx^\alpha \in (\BC^{n})^* \cdot \left( q_nS^d \BC^{n} \right)$ for every monomial $\bfx^\alpha \in S^{d-1}\BC^{n}$. This shows that every non-square-free monomial belongs to $(\BC^{n})^* \cdot (q_n S^d\BC^{n})$.

Now let $\bfx^\beta \in S^{d+1}\BC^{n}$ be a square-free monomial; suppose $\beta_1 = 1$ and let $\gamma$ be the multi-index with $\gamma_1 = 0$ and $\gamma_j = \beta_j$ for $j\geq 2$, so that $\bfx^\gamma = \bfx^\beta / x_1$. Then 
\[
2\bfx^\beta = \frac{\partial}{\partial x_1} (q_n \bfx^\gamma).
\]
This concludes the proof.
\end{proof}

\begin{theorem}[Reznick, \cite{MR1096187}]\label{thm: flattenings of qk hit everything} 
The space of $e$-th partial derivatives  of $f_{n.k}$ consists of all multiples of $q_n^{k-e}$ of degree $2k-e$. In symbols, for all  $n,k$ and $e\leq k$,   $(f_{n,k})_{e,2k-e}(S^e (\BC^{n})^*)= q_n^{k-e} S^e \BC^{n}$. In particular, for all $n,k,e$, the flattening $(f_{n,k})_{e,2k-e}$ has full rank.
\end{theorem}
 
\begin{proof} 
 We proceed by induction on $e$. For $e = 0$ there is nothing to prove. Let $e\geq 1$. By the induction hypothesis, the $(e -1)$-st flattening surjects onto $ q_n^{k-e+1} S^{e-1} \BC^{n}$. It suffices to show that 
 \[
  (\BC^{n})^* \cdot (q_n^{k-e+1} S^{e-1} \BC^{n}) = q_n^{k-e} S^e \BC^n.
 \]
Notice that, for a monomial $\bfx^\alpha \in S^{e-1} \BC^{n}$, 
\begin{align*}
\frac{\partial}{\partial x_i} (q_n^{k-e+1} \bfx^\alpha ) &= 2(k-e+1) x_i q_n^{k-e} \bfx^\alpha  + q_n^{k-e+1} \frac{\partial \bfx^\alpha }{\partial x_i}  \\ 
&= q_n^{k-e} \left(2(k-e+1) x_i  \bfx^\alpha  + q_n \frac{\partial \bfx^\alpha}{\partial x_i}\right).
\end{align*}
Up to rescaling $q_n$ and the differential operators, the term in parenthesis is $\frac{\partial}{\partial x_i} (\bfx^\alpha q_n)$. These terms span $S^{e}\BC^{n}$ by Proposition \ref{prop: hit everything with deg one differentials}.
\end{proof}

\begin{remark}\label{rmk: Lemma holds for every quadric}
The results of Proposition \ref{prop: hit everything with deg one differentials} and Theorem \ref{thm: flattenings of qk hit everything}
hold for every  non-degenerate quadratic form $g \in S^2 \bbC^n$. Indeed every quadric can be diagonalized by the action of $GL_n$ and the proof only uses the fact that $q_n$ is diagonal. In particular, for every $g \in S^2 \bbC^n$, and every $e \leq k$, we have
\[
( S^e \bbC^n)^* \cdot (g^k ) = g^{k-e} S^{e} \bbC^n.
\]
\end{remark}

In order to prove the analog of Prop. \ref{prop: hit everything with deg one differentials} for $\tilde{f}_{n,k}$, we will exploit the decomposition of $S^d \bbC^n$ into Specht modules under the action of symmetric group $\frakS_n$ which permutes the variables. In particular, if $\bbC^n = [n] \oplus [n-1,1]$ where $[n] = \langle \ell_n \rangle$ is an invariant subspace and $[n-1,1] = \langle x_i - x_j : i,j = 1 \vvirg n\rangle$ is isomorphic to the standard representation of $\frakS_n$.   In particular, a homogeneous polynomial $f \in S^d \bbC^n$ can be written as 
\begin{equation}\label{eqn: bigraded decomp}
 f = \sum_{e = 0}^d \ell ^e g_{d-e}
\end{equation}
where $g _{d-e} \in S^{d-e} [n-1,1]$; this is the decomposition of $f$ as sum of bi-homogeneous terms according to the decomposition $\bbC^n = [n ] \oplus [n-1,1]$.

We redefine our basis of $\bbC^n$ and of $\bbC^{n*}$ in accordance to the splitting $[n] \oplus [n-1,1]$: consider $[n] = \langle \ell_n\rangle$ and $[n-1,1] = \langle x_i - x_1 : i = 2 \vvirg n\rangle$ and in the dual space $\ell_n^* = \frac{1}{n} (\frac{\partial }{\partial x_1 } + \cdots + \frac{\partial }{\partial x_n } )$ and $\Delta_{i} = \frac{1}{2} (\frac{\partial}{\partial x_i} -\frac{\partial}{\partial x_1})$. In particular, $\langle \Delta_i : i = 2 \vvirg n\rangle \simeq [n-1,1]$ and $\langle \ell_n^* \rangle \simeq [n]$. Since $\Delta_i ( \ell_n ) = 0$ and $\ell_n^* \cdot [n-1,1] = 0$, this gives the splitting $\bbC^{n*} = \langle \ell_n \rangle ^* \oplus [n-1,1]^*$, where we identify $\langle \ell_n \rangle ^*$ and $[n-1,1]^*$ as subspace of $\bbC^n$.

\begin{remark}\label{rmk: write q as ell plus g}
It will be useful to write $q_n = x_1 ^2 + \cdots + x_n^2$ and its powers as in eqn. \eqref{eqn: bigraded decomp}. Notice that $S^2 \bbC^n$ contains a two-dimensional space of $\frakS_n$-invariants. To see this, consider the decomposition (see e.g. \cite[eqn. (6.7.1)]{Lan:TensorBook}) $S^2 \bbC^n = S^2 (\langle \ell_n\rangle \oplus [n-1,1] ) = S^2 \langle \ell_n \rangle \oplus \langle \ell_n \rangle \otimes [n-1,1] \oplus S^2 [n-1,1]$:  the subspace $S^2 \langle \ell_n \rangle$ is one-dimensional and $\frakS_n$ acts trivially on it, so this is a space of a invariants generated by $\ell_n^2$; the space $ \langle \ell_n \rangle \otimes [n-1,1]$ is isomorphic to $[n-1,1]$, which is irreducible, so it contains no invariants; the subspace $ S^2 [n-1,1]$ contains a one-dimensional subspace of invariants, generated by $g_n = \frac{1}{2}\sum_{i,j} (x_i - x_j)^2$. The uniqueness of the invariant $g_n$ in $S^2 [n-1,1]$ follows by Schur's Lemma: indeed $S^2 [n-1,1] \subseteq [n-1,1] \otimes [n-1,1] \simeq [n-1,1]^* \otimes [n-1,1] = \End ([n-1,1])$ because 
Specht modules are self dual (see e.g. \cite[Ch. 4]{FulHar:RepTh}); by Schur's Lemma, the only $\frakS_n$-equivariant endomorphism of the irreducible representation $[n-1,1]$ is the identity (up to scale); we deduce that $S^2 [n-1,1]$ contains at most a one-dimensional space of invariants.

Since $q_n$ is an $\frakS_n$-invariant, we deduce that $q_n$ is a linear combination of $g_n$ and $\ell_n^2$: indeed, $q_n = \frac{1}{n}( \ell_n^2 + g_n)$.
\end{remark}

We can now prove the following result. 

\begin{theorem}\label{thm: flattenings of ellqk hit everything}
 For every $n,k,e$, the flattening $(\tilde f_{n,k})_{e,2k+1-e}$ has full rank.
\end{theorem}
\begin{proof}

As before, write $\bbC^n = \langle \ell_n \rangle \oplus [n-1,1]$. Up to rescaling $q_n$, write $q_n = \ell_n^2 + g_n$ as in Remark \ref{rmk: write q as ell plus g}. We have 
\[
 \tilde{f}_{n,k} = \ell_n q_n^k = \ell_n (\ell_n^2 + g_n)^k = \ell_n \textsum_0^k \textbinom{k}{j}\ell_n^{2(k-j)} g_n^{k-j} = \textsum_0^k \ell_n^{2(k-j)+1} g_n^{j} 
\]

Let $e \leq k$. In order to calculate $\rk (\tilde{f}_{e,2k+1-e})$, we study the image of $\tilde{f}_{e,2k+1-e}$ as an $\frakS_n$-module. We have $S^e \bbC^{n*} = S^e ( \langle \ell_n^* \rangle \oplus [n-1,1]^* )  =   \bigoplus_{j = 0}^e \langle (\ell_n^*)^j \rangle \otimes S^{e-j} [n-1,1]^* $, so   
\begin{equation}\label{eqn: decomposition of image of ellqk}
\begin{aligned}
 S^e ( \langle \ell_n^* \rangle \oplus [n-1,1]^* ) \cdot (\ell_n q_n^k) &= \biggl( \bigoplus_{j = 0}^e \langle (\ell_n^*)^j \rangle \otimes  S^{e-j} [n-1,1]^* \biggr) \cdot (\ell_n q_n^k)  \\
&= \sum_{j = 0}^e \langle (\ell_n^*)^j\rangle   (S^{e-j} [n-1,1]^* \cdot (\ell_n q_n^k)) .
\end{aligned}
\end{equation}
To conclude, we show that the last summation is in fact a direct sum. In the $j$-th summand, we have 
\begin{equation}\label{eqn: summation in thm ell q}
\begin{aligned}
\langle (\ell_n^*)^j\rangle   \cdot (S^{e-j} [n-1,1]^* \cdot (\ell_n q_n^k)) &=  S^{e-j} [n-1,1]^* \cdot \left[ \ell_n^{*j} \cdot  \left( \textsum_{i=0}^k \textbinom{k}{i} \ell_n^{2i+1} g_n^{k-i}\right)\right] =  \\ 
&= S^{e-j} [n-1,1]^* \cdot  \biggl[ \sum_{i=\lceil \frac{j-1}{2}\rceil}^k \textbinom{k}{i} \ell_n^{2i+1 - j} g_n^{k-i}\biggr] .
\end{aligned}
\end{equation}

If $\Delta \in S^{e-j} [n-1,1]^* $, we have $\Delta \left( \sum_{i=\lceil \frac{j-1}{2}\rceil}^k \textbinom{k}{i} \ell_n^{2i+1 - j} g_n^{k-i}\right) = \sum_{i=\lceil \frac{j-1}{2}\rceil}^{k} \textbinom{k}{i} \ell_n^{2i+1 - j} \Delta( g_n^{k-i})$. In particular the summands with $i > k-e+j$ are $0$; moreover, the $i$-th term of the summation, varying $\Delta$, ranges on the entire $\ell_n^{2i+1-j} g_n^{k-i-(e-j)} S^{e-j} [n-1,1]$ by Theorem \ref{thm: flattenings of qk hit everything} and Remark \ref{rmk: Lemma holds for every quadric}. We deduce that the leading term (in $\ell_n$) of the last line of \eqref{eqn: summation in thm ell q} ranges on the whole $\ell_n^{2(k-e+j) +1 -j } S^{e-j} [n-1,1] = \ell_n^{2k-2e+j -1} S^{e-j} [n-1,1]$. This shows that the summands in \eqref{eqn: decomposition of image of ellqk} are linearly independent because their leading terms have different degrees, the $j$-th one having degree in $\ell_n$ equal to $2k-2e+j -1$.

From \eqref{eqn: decomposition of image of ellqk}, we deduce that rank of $\tilde{f}_{e,2k+1-e}$ is 
\begin{align*}
 \dim (S^e \bbC^{n*} \cdot \tilde{f}_{e,2k+1-e} ) &= \sum_{j = 0}^e \dim S^{e-j} [n-1,1] = \\ &= \sum_{j=0}^e \binom{e-j+n-2}{e-j} = \\ &=\sum_{j=0}^e \binom{n-2+j}{j} = \binom{e+n-1}{e} = \dim S^e \bbC^{n*}.
\end{align*}
hence, $\tilde{f}_{e,2k+1-e}$ is injective and this conclude the proof when $e \leq k$. Since $f_{e,2k+1-e}$ is the transpose of $f_{2k+1-e,e}$, the proof is complete.
\end{proof}

\section{Two auxiliary results}

  \begin{proposition} \label{vpeprop}
 Let $n(m),k(m)$ be   polynomially bounded functions of $m$. Then the sequences $\{f_{n,k}\}_m$ and $\{\tilde f_{n,k}\}_m$ are   in the algebraic complexity class $\bfVP_e$ of sequences admitting polynomial size formulas.
 \end{proposition}

\begin{proof}
The standard expression
\begin{align*}
f_{n,k} &= q_n \cdots q_n \quad \text{($k$ times), with} \\
q_n &= x_1 \cdot x_1 + \cdots + x_n \cdot x_n
\end{align*}
gives a formula of size $k(3n-1)$ for $f_{n,k}$. The additional $\ell_n = x_1 + \cdots + x_n$ and $\tilde{f}_{n,k} = \ell_n \cdot f_{n,k}$ provides a formula of size $k(3n-1)+n$ for $\tilde{f}_{n,k}$.
\end{proof}

\begin{proposition}\label{prop: IMM specializes to fnk}
If $n = 2m$ (resp. $n = 2m +1$) then the polynomial $f_{n,k}$ is a specialization of the matrix powering
polynomial $\Pow^{2k}_{m+1}$ (resp. $\Pow^{2k}_{m+2}$).
If $n = 2m$ (resp. $n = 2m +1$) then the polynomial $\tilde{f}_{n,k}$ is a specialization of the iterated matrix multiplication
polynomial $\IMM^{2k+1}_{m+2}$ (resp. $\IMM^{2k+1}_{m+3}$).
\end{proposition}
\begin{proof}
Let $n = 2m+1$ and set $y_j^\pm = x_{2j-1} \pm \sqrt{-1} x_{2j}$ for $j = 1 \vvirg m$. Consider the specialization of $\Pow^{2k}_{m+2}$ to the matrix
\[
 Q_m = \left( \begin{array}{ccccc}
  0 & y_1^+ & \cdots & y_m^+ & x_n \\
  y_1^- & & & & \\
  \vdots &  &\bf0 & &\\
  y_m^- & & & & \\
  x_n & & & & 
 \end{array}\right),
\]
of size $m+2$.

We show that $\Pow^{2k}_{m+2}(Q_m) = f_{n,k}$, up to scale. The characteristic polynomial of $Q_m$ is
\[
 \det(Q_m - t\Id_{m+2}) = (-1)^{m+2} t^{m}\left( t^2 - \textsum_j y_j^+ y_j^- - x_n^2\right) = (-1)^{m+2} t^{m}\left( t^2 - q_n\right).
\]
Thus, the nonzero eigenvalues of $Q_m$ (as functions of $x_1\vvirg x_n$),  are $\pm \sqrt{q_n}$. In particular, $\Pow^{d}_{m+2}(Q_m) = (\sqrt{q_n})^d + (-\sqrt{q_n})^d$ is $0$ if $d$ is odd and it is $2 q_n^k = 2f_{n,k}$ if $d = 2k$ is even.

If $n = 2m$ is even, apply the same argument to the matrix obtained from $Q_m$ by removing the last row and the last column.

Similarly, $\tilde f_{n,k}$ is a specialization of $\IMM^{2k+1}_{m+2}$ or $\IMM^{2k+1}_{m+3}$ (depending on the parity of $n$) by making the first matrix $\ell_n\Id$ and specializing the remaining matrices to the matrix above.
\end{proof}

\section{Proof of Theorem \ref{noshiftthm}}

Since $ n \geq m^5$, we may choose a linear inclusion $\BC^{m^2}\subseteq \BC^{mn^2}$ and regard   $\perm_m\in S^m \bbC^{mn^2}$. Our goal is to show that for every $s,\tau$
\[
\dim \langle \partial^{=s} \IMM^m_n \rangle_{=\tau} \geq \dim \langle \partial^{=s} \perm_m \rangle_{=\tau}.
\]

We split the proof into three cases. In the first and in the second case, we degenerate $\IMM^m_n$ to $f_{n,k}$ if $m=2k$ is even and to $\tilde{f}_{n,k}$ if $m=2k+1$ is odd. This is possible by Proposition \ref{prop: IMM specializes to fnk}. Write $F_{m,n}$ for either $f_{n,k}$ or $\tilde f_{n,k}$ in what follows. Since $\IMM^m_n$ degenerates to $F_{m,n}$, we have $\dim \langle \partial^{=s} \IMM^m_n \rangle_{=\tau} \geq \dim \langle \partial^{=s} F_{m,n} \rangle_{=\tau}$. 

In the third case, we specialize $\IMM^m_n$ to the power sum polynomial of degree $m$ in $m^2$ variables $y_1^{m} + \cdots + y_{m^2}^m$ by specializing every argument of $\IMM^m_n$ to the diagonal matrix of size $n \times n$ with $y_1 \vvirg y_{m^2}$ in the first $m^2$ diagonal entries and $0$ elsewhere.

{\bf Case 1}:  $s \geq \lceil \frac{m}{2}\rceil $. We show that $\dim \langle \partial^{=s} F_{m,n} \rangle_{=\tau} \geq \dim \langle \partial^{=s} \perm_m \rangle_{=\tau}$ when $s\geq\lceil\frac m2\rceil$. Up to the action of $GL_{mn^2}$, assume $\bbC^{m^2} \subseteq \bbC^n \subseteq \bbC^{mn^2}$, where $\bbC^{m^2}$ is the space spanned by the variables of $\perm_m$ and $\bbC^n$ is the space spanned by the variables of $F_{m,n}$. It will suffice to prove $\dim  \langle \partial^{=s}(\perm_m \in S^m\BC^{n})\rangle_{=\tau}\leq \dim  \langle \partial^{=s}F_{m,n} \in S^m\BC^{n}\rangle_{=\tau}$ because the remaining $mn^2-n$ variables will contribute the same growth to the ideals $\calJ_s(\perm_m)$ and $\calJ_s(F_{m,n})$. Since $s\geq \lceil \frac{m}{2}\rceil$, $(F_{m,n})_{s, m-s}$ surjects onto $S^{m-s} \bbC^n$ by Theorem \ref{thm: flattenings of qk hit everything} and Theorem \ref{thm: flattenings of ellqk hit everything}, and thus, for every shift $\tau$, the shifted partial derivative map surjects onto $S^{ m-s +\tau}
\BC^n$ for all $\tau$. This shows $\langle \partial^{=s}F_{m,n} \in S^m\BC^{n}\rangle_{=\tau} = S^{ m-s +\tau}\BC^n$ and 
proves this case.

{\bf Case 2}: $s <\lceil \frac m2\rceil $ and $\tau< 2m^3$. Again, it suffices to prove $\dim  \langle \partial^{=s} \perm_m\in S^m\BC^{n} \rangle_{=\tau}\leq \dim  \langle \partial^{=s}F_{m,n} \in  S^m\BC^{n} \rangle_{=\tau}$. Since $s \leq \lceil m/2 \rceil$, $(F_{m,n})_{s,m-s}$ is injective, and its partials have image of dimension $\binom{n+s-1}{s}$. Corollary 2.4 of \cite{2016arXiv160902103E} states that   any subspace of $S^{m-s}\BC^n$ of dimension   $\binom{n+s-1}s$ generates
an ideal that in degree $m-s+\tau$ has dimension at least $\binom {n+s+\tau-1}{s+\tau}$; this is a consequence of a general result on the growth of ideals known as Macaulay's Theorem (see e.g. \cite{Stan:HilbertFunGradedAlgebras}). Thus
\[
 \dim  \langle \partial^{=s}F_{m,n} \in  S^m\BC^{n} \rangle_{=\tau} \geq \binom {n+s+\tau-1}{s+\tau},
\]
(and equality holds in the case $m$ is even).

We compare this with the crude estimate for $\perm_m$ that ignores syzygies of its $s$-th Jacobian ideal. The space $\langle \partial^{=s} \perm_m\subseteq S^m\BC^n \rangle_{=0}$ has dimension $\binom{m}{s}^2$, because $s$-th partial derivatives of $\perm_m$ are subpermanents of size $m-s$, so there is one for every choice of $s$ rows and $s$ columns of the matrix. Ignoring syzygies of $\calJ_s(\perm_m)$,
\[
{\binom m s}^2\binom{n+\tau-1}{\tau} \geq \dim  \langle \partial^{=s}(\perm_m\subseteq S^m\BC^n) \rangle_{=\tau}.
\]
We will conclude that $\binom {n+s+\tau-1}{s+\tau}>{\binom m s}^2\binom{n+\tau-1}{\tau}$ in the range we consider. This  is equivalent to
\begin{equation}\label{Nestimate}
\frac{(n+s+\tau-1)( n+s+\tau-2)\cdots (n+\tau)}{(\tau+s)(\tau+s-1)\cdots (\tau+1)}>{\binom{m}{s}}^2.
\end{equation}
The left hand side is bounded from below by $\frac{(n+\tau)^s}{(\tau + s)^s}$ and the right hand side is bounded from above by $m^{2s}$, so that a sufficient condition for \eqref{Nestimate} is 
 \[
  \frac{n+\tau}{\tau+s} > m^2.
 \]
which holds when $n \geq m^5$  and $\tau \leq 2m^3$ as $s \leq m$.

{\bf Case 3}: $s <\frac{m}{2}$ and $\tau>m^3$. Here set all matrices $(X_1 \vvirg X_m)$ equal to a matrix that is  zero except for  the first $m^2$ entries on the diagonal, call them $y_1 \vvirg y_{m^2}$. The resulting degeneration of $\IMM^m_n$  is $y_1^m+\cdots + y_{m^2}^m$. As in Case 1, it will suffice to prove the result for  the shifted partials of both polynomials in $m^2$ variables because the remaining $mn^2-m^2$ variables will contribute the same growth to both  ideals. The space of partial derivatives of order $s$ is $\langle y_1^{m-s} \vvirg y_{m^2}^{m-s}\rangle$. The image of the $\tau$-th shifted partial map consists of all polynomials in $m^2$ variables of degree $m-s+\tau$ as soon as $m-s+\tau> m^2(m-s-1)+1$, so that every monomial of degree $m-s+\tau$ is divisible by at least one power of order $m-s$. In particular, the shifted partials derivative map is surjective whenever when $\tau>m^3$.

\section{Complete symmetric functions}

 Recall that  $h_{n,d}$ is the complete symmetric function of degree $d$ in $n$ variables:
\[
 h_{n,d} = \sum_{\vert \alpha \vert = d} \bfx^\alpha,
\]
 where the summation is over all multi-indices $\alpha = (\alpha_1 \vvirg \alpha_n)$ with $\alpha_1 + \cdots + \alpha_n = d$.

\begin{proposition}\label{prop: derivatives of h}
 For every monomial $\bfx^\beta$ with $\vert \beta \vert = e \leq d$, we have 
 \[
  \frac{\partial^e}{\partial \bfx^\beta} h_{n,d} = \beta! \cdot h_{n+e , d-e} (x_1 \vvirg x_n, \bfx^{(\beta)}),
 \]
where $\bfx^{(\beta)} = (\underbrace{x_1 \vvirg x_1}_{\beta_1} \vvirg \underbrace{x_n \vvirg x_n}_{\beta_n})$ and $\beta! = \beta_1! \cdots \beta_n!$. In particular the image of the flattening $(h_{n,d})_{e,d-e}$ is the space
\[
\left\langle  h_{n+e,d-e}(x_1 \vvirg x_n,\bfx^{(\beta)}) : \vert \beta \vert = e \right\rangle \subseteq S^{d-e} \bbC^n.
\]
\end{proposition}
\begin{proof}
 We proceed by induction on $e$. If $e = 1$, suppose $\bfx^\beta = x_n$ and write $h_{n,d} = \sum_{j=0}^d x_n^j h_{n-1, d-j}(x_1 \vvirg x_{n-1})$, so that
 \begin{align*}
  \frac{\partial}{\partial x_n} h_{n,d} &= \sum_{j=0}^d j \cdot x_n^{j-1} h_{n-1, d-j}(x_1 \vvirg x_{n-1}) = \\
  &= \sum_{\ell=0}^{d-1} (\ell+1) \cdot x_n^\ell h_{n-1, d-\ell -1}(x_1 \vvirg x_{n-1}) = \\
  & = \left[\sum_{\ell=0}^{d-1} x_n^\ell h_{n-1, d-\ell -1}(x_1 \vvirg x_{n-1})\right] + x_n\left[\sum_{\ell=0}^{d-2} (\ell+1) \cdot x_n^\ell h_{n-1, d-\ell -2}(x_1 \vvirg x_{n-1})\right]
 \end{align*}
where the first summation in the last line contains one term from each summand in the previous line, and the second summation contains the remaining terms (with shifted indices). The first summation adds up to $h_{n,d-1} (x_1 \vvirg x_n)$; by repeating this on the second summation we obtain
\[
h_{n,d-1} (x_1 \vvirg x_n) + x_n h_{n,d-2} (x_1 \vvirg x_n) + x_n^2 \left[\sum_{\ell=0}^{d-3} (\ell+1) \cdot x_n^\ell h_{n-1, d-\ell -3}(x_1 \vvirg x_{n-1})\right]
\]
and iterating this process we obtain $\sum_{j = 0}^d x_n^j h_{n,d-1 -j} (x_1 \vvirg x_n) = h_{n+1,d-1}(x_1 \vvirg x_n,x_n)$ proving the base case.

Let $e \geq 1$ and suppose $\beta_1 \geq 1$. Let $\gamma = (\beta_1-1,\beta_2 \vvirg \beta_n)$. We have
 \begin{align*}
  \frac{\partial^e}{\partial \bfx^\beta} h_{n,d} &= \frac{\partial }{\partial x_1} \frac{\partial ^{e-1}}{\partial \bfx^\gamma} h_{n,d} = \\
  &= \gamma! \cdot \frac{\partial}{\partial x_1} h_{n+e-1,d-e+1}(x_1 \vvirg x_n, \bfx^{(\gamma)}).
  \end{align*}
 By chain rule and by symmetry
  \begin{align*}
\frac{\partial}{\partial x_1} &h_{n+e-1,d-e+1}(x_1 \vvirg x_n, \bfx^{(\gamma)}) = \\
&=(\gamma_1+1) \frac{\partial}{\partial y_1}\bigl\vert_{(x_1 \vvirg x_n, \bfx^{(\gamma)},x_1)} \bigr. h_{n+e-1,d-e+1}(y_1 \vvirg y_{n+e-1}) = \\ 
&=\beta_1 h_{n+e,d-e}((x_1 \vvirg x_n, \bfx^{(\gamma)},x_1)),
  \end{align*}
  where we used the case $e = 1$ again. Since $\gamma! \cdot \beta_1 = \beta!$, we conclude.
\end{proof}

\begin{proposition}\label{prop: coefficients of monomials is expanded h}
For any choice of multi-indices $\beta,\gamma$ with $\vert \beta \vert = p$ and $\vert \gamma \vert = e$, the coefficient of $\bfx^\gamma$ in $h_{n+p,e} ( x_1 \vvirg x_n, \bfx^{(\beta)})$ is 
\[
 \prod_{i=1}^n \binom{\beta_i + \gamma_i}{\gamma_i}.
\]
\end{proposition}
\begin{proof}
Write $[f]_{\gamma}$ for the coefficient of $\bfx^\gamma$ in the polynomial $f$.

We use induction on $p$. If $p = 0$, then $h_{n+p,e} ( x_1 \vvirg x_n, \bfx^{(\beta)}) = h_{n,e}$ and for every $\gamma$ we have $[h_{n,e}]_\gamma = 1 =  \textprod_{i=1}^n \binom{\gamma_i}{\gamma_i}$.

Let $p\geq 1$ and suppose $\beta_1 \geq 1$. Write $h_{n+p,e}(\bfy) = \sum_{j=0}^e y_1^{j} h_{n+p-1,e-j}(y_2 \vvirg y_{n+p})$. Let $\eta_1 = (1 ,0\vvirg 0) \in \bbZ^{n}_{\geq 0}$. We have
\[
 [h_{n+p,e}(x_1 \vvirg x_n,\bfx^{(\beta)})]_\gamma = \sum_{j=0}^{\gamma_1} [h_{n+p-1,e-j}(x_1 \vvirg x_n, \bfx^{(\beta-\eta_1)})]_{\gamma-j\eta_1}.
\]
Apply the inductive hypothesis to the summands of the right hand side  to get 
\begin{align*}
 [h_{n+p,e}(x_1 \vvirg x_n,\bfx^{(\beta)})]_\gamma &= \sum_{j=0}^{\gamma_1} \left[\binom{\gamma_1 - j + \beta_1 -1 }{\gamma_1 -j} \cdot \prod_{i=2}^n \binom{\beta_i + \gamma_i}{\gamma_i}\right]  = \\ 
 &= \left( \sum_{j=0}^{\gamma_1} \binom{\beta_1 -1 +j}{j}\right) \prod_{i=2}^n \binom{\beta_i + \gamma_i}{\gamma_i}  =\\
 &= \prod_{i=1}^n \binom{\beta_i + \gamma_i}{\gamma_i}
\end{align*}
\end{proof}

Proposition \ref{prop: coefficients of monomials is expanded h} shows that the entries of the matrix representing the partial derivatives map of $h_{n,d}$ in the monomial basis are products of binomial coefficients with a special combinatorial structure. Matrices with this structure are the object of study of the Lindstr\"om-Gessel-Viennot theory on totally nonnegative matrices (see e.g. \cite{MR815360}). We will provide some results on matrices with this structure which will be used to prove that the matrices described in Proposition \ref{prop: coefficients of monomials is expanded h} have full rank.

Let $a_1 \vvirg a_N \in \bbZ_{\geq 0}$ be nonnegative integers and let $G(a_1 \vvirg a_N)$ be the $N \times N$ symmetric matrix whose $(i,j)$-th entry is $\binom{a_i+a_j}{a_i}$. The Lindstr\"om-Gessel-Viennot Lemma (see \cite[\S2]{MR815360}) guarantees that $G(a_1 \vvirg a_N)$ is a totally nonnegative matrix (in the sense that every minor is nonnegative), and its rank is equal to the number of distinct $a_i$'s. In particular, $G(a_1 \vvirg a_N)$ is always positive semidefinite and it is positive definite if and only if the $a_i$'s are distinct. Moreover if $a_{i_1} = a_{i_2}$ for some $i_1,i_2$, then the $i_1$-th and $i_2$-th rows are equal. 

Given two matrices $A,B$ of the same size, define $A \odot B$ to be the \emph{Hadamard product} of $A$ and $B$. For vectors $\bfa_1 \vvirg \bfa_N \in \bbZ^{m}_{\geq 0}$, with $\bfa_i = (a_{ij}) _{j = 1 \vvirg m}$ define 
\[
 G(\bfa_1 \vvirg \bfa_N) := \bigodot_{i=1}^m G(a_{1,i}\vvirg a_{N,i}).
\]
Our goal is to prove that $G(\bfa_1 \vvirg \bfa_N)$ is positive definite if the  $\bfa_i$ are distinct.

We will need the following two technical results.  Given a matrix $A$ we denote by $A^i_\bullet$ (resp. $A^\bullet_i$) the $i$-th row (resp. column) of $A$ and by $A^I_J$ the submatrix consisting of rows in the set of indices $I$ and columns in the set of indices $J$. 

\begin{lemma}[\cite{266847}]\label{lemma: principal minor on indep rows}
 Let $A$ be symmetric, positive semidefinite. Let $I = \{ i_1 \vvirg i_r\}$ be a set of indices such that the $r$ vectors $\{ A^\bullet_i\}_{i \in I}$ are linearly independent. Then the principal submatrix $A^I_I$ of $A$ has full rank.
\end{lemma}
\begin{proof}
 Without loss of generality, suppose $I = \{ 1 \vvirg r\}$ and let $R = A^I_I$. We want to prove that $R$ is full rank, namely that $Ru = 0$ for some $u \in \bbR^r$ implies $u =0$. Let $v \in \bbR^n$ such that $v_i = u_i$ if $i\leq r$ and $v_i = 0$ if $i > r$. Since $A$ is positive semidefinite, write $A = B^T B$. We have 
 \[
0 =  u^T R u = v^T A v = v^T B^T B v = \Vert Bv \Vert.
 \]
In particular $Bv = 0$, therefore $Av =0$; since the first $r$ columns of $A$ are linearly independent, we deduce $v = 0$, so that $u =0$ and $R$ is nonsingular.
\end{proof}

\begin{lemma}\label{lemma: posdef odot positive psd is pos def}
Let $A,B$ be symmetric $N \times N$ matrices such that $A$ is positive definite and $B$ is positive semidefinite with strictly positive diagonal entries. 
Then $A \odot B$ is positive definite.
\end{lemma}
\begin{proof}
Given a symmetric matrix $C$, denote by $C_{(k)}$ the $k$-th leading principal submatrix of $C$, namely the $k\times k$ submatrix consisting of the first $k$ rows and $k$ columns of $C$. Then $C$ is positive definite if and only if $\det(C_{(k)})$ is positive for every $k$. Therefore it suffices to show that the leading principal minors of $A \odot B$ are positive.

Let $k_0$ be the smallest $k$ such that $\det(B_{(k)}) = 0$. From Schur's Product Theorem (see, e.g. \cite[Ex.~36.2.1]{MR1277174}), the Hadamard product of positive definite matrices is positive definite. For every $k < k_0$, we have that $B_{(k)}$ is positive definite, so $(A\odot B)_{(k)} = A_{(k)} \odot B_{(k)}$ is positive definite as well and in particular it has positive determinant.

If $k \geq k_0$, from \cite[Eqn.~1.11]{Lynn:schurProductInequalities}, we have 
 \[
  \det(A_{(k)} \odot B_{(k)}) + \det(A_{(k)}B_{(k)}) \geq \det(A_{(k)}) \textprod_{i=1}^{k} b_{ii} + \det(B_{(k)}) \textprod_{i=1}^{k} a_{ii}.
 \]
Now, for $k \geq k_0$, $B_{(k)}$ is a positive semidefinite matrix and $B_{(k_0)}$ is a principal submatrix with $\det( B_{(k_0)})= 0$, so $\det( B_{(k)}) = 0$ as well. Therefore
 \[
 \det((A \odot B)_{(k)} ) \geq \det(A_{(k)}) \textprod_{i=1}^k b_{ii} > 0
\]
proving that $(A \odot B)_{(k)} $ has positive determinant.  
\end{proof}

\begin{proposition}\label{prop: hadamard product of many gessel viennot}
Let $\bfa_1 \vvirg \bfa_N \in \bbZ^m_{\geq 0}$. Then the rank of $G(\bfa_1 \vvirg \bfa_N)$ equals the number of distinct $m$-tuples $\bfa_1 \vvirg \bfa_N$. 
\end{proposition}
\begin{proof}
We proceed by induction on $m$. For $m=1$, the statement follows from the Lindstr\"om-Gessel-Viennot Lemma.

If $m\geq 2$, for every $i=1 \vvirg N$, write $\bfa_i = (\bfa_i',a_{i,m})$. Let $A = G(\bfa_1' \vvirg \bfa_N')$ and $B= G(a_{m,1} \vvirg a_{m,N})$. By the induction hypothesis $A$ is positive semidefinite and its rank is equal to the number of distinct $\bfa_i'$'s. Similarly for $B$.

Let $C =  G(\bfa_1 \vvirg \bfa_N) = A \odot B$. If two pairs $\bfa_i = \bfa_j$, then $A^\bullet_i = A^\bullet_j$ and $B^\bullet_i = B^\bullet_j$ so that the corresponding two columns of $C$ are equal. Conversely, if two columns $C^\bullet_i,C^\bullet_j$ are equal, we show that the corresponding $m$-tuples $\bfa_i$ and $\bfa_j$ are equal. Consider the principal $2\times2$ submatrix obtained from these two columns:
\[
 C^{ij}_{ij} = A^{ij}_{ij} \odot B^{ij}_{ij}.
\]
If $A^\bullet_i \neq A^\bullet_j$, then they are linearly independent by the induction hypothesis, and by Lemma \ref{lemma: principal minor on indep rows} the submatrix $A^{ij}_{ij}$ is positive definite. The submatrix $B^{ij}_{ij}$ is positive semidefinite and has strictly positive diagonal entries, therefore by Lemma \ref{lemma: posdef odot positive psd is pos def} $C^{ij}_{ij}$ is positive definite, in contradiction with the assumption. This shows that if $C^\bullet_i = C^\bullet_j$, then $A^\bullet_i = A^\bullet_j$ and therefore $B^\bullet_i = B^\bullet_j$. In particular $\bfa_i = \bfa_j$.

Therefore, we may assume that $C$ has distinct columns and our goal is to show that $C$ has full rank. Suppose by contradiction that $C$ does not have full rank and let
\begin{equation}\label{eqn: vanishing linear combination of C}
0 = \alpha_1  C^\bullet_1 + \cdots + \alpha_N C^\bullet_N
\end{equation}
be a vanishing linear combination of the columns of $C$.

Up to conjugation by a permutation matrix, suppose there exist $0 = k_0 < k_1 < \cdots < k_r=N$ such that $\bfa_i' = \bfa_j'$ if $k_s < i,j \leq k_{s+1}$ and $\bfa_i' \neq \bfa_j'$ otherwise. Notice that if the $\bfa_i$'s are distinct, then $A$ has full rank and so does $C$ from Lemma \ref{lemma: posdef odot positive psd is pos def} because $B$ is positive semidefinite with strictly positive entries. Therefore, suppose $k_1 \geq 2$ and up to reducing to a principal submatrix suppose that $\alpha_{k_1} \neq 0$. Since the first $k_1$ columns (and rows) of $A$ are equal, the first $k_1$ columns (and rows) of $B$ are linearly independent, otherwise two of them would be equal, providing that two $m$-tuples $\bfa_i$ and $\bfa_j$ for $i,j \leq k_1$ would be equal. By Lemma \ref{lemma: principal minor on indep rows}, the principal submatrix $B^{1\vvirg k_1}_{1 \vvirg k_1}$ is positive definite.

The linear combination \eqref{eqn: vanishing linear combination of C} can be written as 
\[
 0 = A^\bullet_1 \odot(\alpha_1 B^\bullet_1 + \cdots + \alpha_{k_1} B^\bullet_{k_1}) + \textsum_{i > k_1} \alpha_i (A^\bullet_i \odot B^\bullet_i).
\]

Define $\tilde{A}$ to be the matrix obtained from $A$ by removing the first $k_1-1$ rows and columns, that is $\tilde{A} = G(\bfa_{k_1}' \vvirg \bfa_N')$. $\tilde{A}$ has the same rank as $A$. Let $B' = P^T B P$, for
\[
 P = \left( \begin{array}{ccccccc}
              1 & & & \alpha_1& & & \\
              & \ddots & & \vdots & & & \\
              & & 1 & \vdots & & & \\
              & & & \alpha_k& & & \\
              & & & & 1 & & \\
              & & & & & \ddots & \\
              & & & & & & 1\\
            \end{array}
 \right);
\]
notice that $B'$ is obtained from $B$ by performing row and column operations. In particular $B'$ has the same signature as $B$; moreover, from the block structure of $P$, we deduce that the submatrix ${B'}^{1 \vvirg k_1}_ { 1\vvirg k_1}$ has the same signature as $B^{1 \vvirg k_1}_ { 1\vvirg k_1}$, namely it is positive definite. This shows that the $k_1$-th diagonal entry of $B'$ is strictly positive. Define $\tilde{B}$ to be the submatrix obtained from $B'$ by removing the first $k_1-1$ rows and columns. Define $\tilde{C} = \tilde{A} \odot \tilde{B}$. The linear combination of \eqref{eqn: vanishing linear combination of C} induces a vanishing linear combination among the columns of $\tilde{C}$. 

By repeating this procedure at most $r$ times, we find a singular $r\times r$ matrix $\tilde{\tilde{C}} =  \tilde{\tilde{A}} \odot \tilde{\tilde{B}}$ with $\tilde{\tilde{A}}$ positive semidefinite and of full rank (so positive definite) and $\tilde{\tilde{B}}$ positive semidefinite with strictly positive diagonal entries. By Lemma \ref{lemma: posdef odot positive psd is pos def}, we obtain a contradiction. This concludes the proof.
\end{proof}

Using these results, we can finally prove:
\begin{theorem}\label{thm: hnd have maximal catalecticant rank}
 For every $n,d,e$, the flattening $(h_{n,d})_{e,d-e} : S^e (\bbC^n)^* \to S^{d-e} \bbC^n$ of the complete symmetric function $h_{n,d}$ has full rank.
\end{theorem}
\begin{proof}
First, we consider the case $d = 2k$ even.

It suffices to prove the result for $e = k$.
\begin{align*}
(h_{n,d})_{k,k} ( S^k (\bbC^n)^*) &= \left\langle \frac{\partial^{\vert \beta \vert}}{\partial \bfx^\beta}h_{n,d} : \vert \beta \vert =k\right\rangle = \\
  &= \left\langle h_{n+k,k}(x_1 \vvirg x_n, \bfx^{(\beta)}) : \vert \beta \vert = k \right\rangle.
\end{align*}
Define two column vectors
\begin{align*}
\bfh &= \bigl(h_{n+k,k}(x_1 \vvirg x_n, \bfx^{(\beta)}) :  \vert \beta \vert = k \bigr)^T \\
\bfb &= \bigl(\bfx^\beta : \vert \beta \vert = k\bigr)^T.
\end{align*}
From Proposition \ref{prop: derivatives of h} and Proposition \ref{prop: coefficients of monomials is expanded h}, we have $\bfh = A \bfb$ where the $(\beta,\gamma)$-th entry of $A$ is 
\[
A_{\beta,\gamma} = \textprod_{i=1}^n\binom{\beta_i+\gamma_i}{\gamma_i}
\]
namely $A = G(\beta : \vert \beta \vert = k)$. Since the multi-indices $\beta$ are all distinct, by applying Proposition \ref{prop: hadamard product of many gessel viennot} we deduce that $A$ is nonsingular and therefore the entries of $\bfb$ are linear combinations of the entries of $\bfh$. This shows that $(h_{n,d})_{k,k}$ is full rank.

Now consider $d = 2k+1$ odd. It suffices to prove the result for $e = k+1$. Let $g = \frac{\partial}{\partial x_1}h_{n,d} = h_{n+1,d-1}(x_1,x_1,x_2 \vvirg x_n) \in S^{2k}\bbC^n$. The image of the flattening $g_{k,k} : S^k (\bbC^n)^* \to S^k \bbC^n$ is contained in the image of  $(h_{n,d})_{k+1,k}$.  
To conclude, we will show  that $g_{k,k}$ is full rank.

Let $h = h_{n+1,d-1}(y_1 \vvirg y_{n+1})$. By the result in the case of even degree, we know that $(h)_{k,k}$ has full rank, namely 
\[
\left\langle h_{n+1+k,k}(y_1 \vvirg y_{n+1}, \bfy^{(\beta)}) : \vert \beta \vert = k\right\rangle = S^k \bbC^{n+1}.
\]
In particular, the image of this space under the specialization $(y_1 \vvirg y_{n+1}) = (x_1,\vvirg x_n,x_1)$ is $S^k \bbC^n$.

On the other hand, notice that,
\[
 h_{n+1+k,k}(y_1 \vvirg y_{n+1}, \bfy^{(\beta)})\vert_{(y_1 \vvirg y_{n+1}) = (x_1 \vvirg x_n,x_1)} = h_{n+k+1,k}(x_1 \vvirg x_n, \bfx^{(\gamma)})
\]
where $\gamma_1 = \beta_1 + \beta_{n+1} + 1$ and $\gamma_i = \beta_i$ for $i = 2 \vvirg n$ (indeed $\vert \gamma \vert = k+1$). This shows that the image of $(h)_{k,k}$ is $S^k \bbC^n$, and therefore $(h_{n,d})_{k+1,k}$ is full rank.
\end{proof}
 
Theorem \ref{thm: hnd have maximal catalecticant rank} implies
\begin{corollary}
For every $n,d$,  
 \[
\bfR_S(h_{n,d}) \geq  \uline{\bfR}_S(h_{n,d}) \geq \binom{n+\lfloor d/2 \rfloor -1}{\lfloor d/2 \rfloor}.
 \]
\end{corollary}
For readers familiar with cactus rank and cactus border rank, by \cite[Thm. 5.3D]{MR1735271}, we obtain the same lower bounds for cactus rank and cactus border rank. 
 
\section{Koszul flattenings}
We recall the definition of Koszul flattening introduced in \cite{MR3081636}. A Koszul flattening map of $p \in S^d \bbC^N$ depends on two parameters $s,q$: it is obtained by tensoring the $s$-th partial derivative map \eqref{flateqn} with the identity map on the space $\Lambda^q \BC^N$ (for some $q$) and then applying the exterior derivative map. In symbols:
\begin{align*}
\begin{array}{clclc}
 \Lambda^q \bbC^N \otimes S^{s} (\bbC^N)^* &\xto{\id \otimes p_{s,d-s}} &\Lambda^q \bbC^N \otimes S^{d-s} \bbC^N & \xto{\delta_{q,d-s}}&\Lambda^{q+1} \bbC^N \otimes S^{d-s-1} \bbC^N \\
 X \otimes D &\xmapsto{\phantom{\id \otimes p_{s,d-s}}} & X \otimes D(p) & &  \\
 & &  X \otimes g & \xmapsto{\phantom{\delta_{q,d-s}}} &\textsum_1^N \left((x_i \wedge X) \otimes \textfrac{\partial g}{\partial x_i}\right)
\end{array}
 \end{align*}

Let $p_{s,d-s}^{\wedge q}:\Lambda^q \BC^N \otimes S^s\BC^{N*} \to \Lambda^{q+1}\BC^N\otimes S^{d-s-1}\BC^N$ denote the composition of the two maps. Explicitly, writing $x_I=x_{i_1}\wedge \cdots \wedge x_{i_{q}}$ for a $q$-tuple $I = (i_1 \vvirg i_{\vert I\vert})$ and $\frac{\partial^{s}  }{\partial x_{J}}=\frac{\partial  }{\partial x_{j_1}}\cdots \frac{\partial  }{\partial x_{j_s}}$ for an $s$-tuple $J = (j_1 \vvirg j_s)$, the map is
\[
p_{s,d-s}^{\wedge q}:  x_{I}\otimes \frac{\partial^{s}  }{\partial x_{J}}
\mapsto x_k\wedge x_I\otimes \frac{\partial^{s+1} p }{\partial x_k \partial x_{J}}.
\]

Then, by \cite[Prop 4.1.1]{MR3081636}, one obtains the border rank lower bound
\begin{equation}\label{eqn: Koszul flat bound for border rank}
\uline{\bfR}_S(p)\geq \frac{\rank(p_{s,d-s}^{\wedge q})}{\rank((\ell^d)_{s,d-s}^{\wedge q})}= \frac{\rank(p_{s,d-s}^{\wedge q})}
{\binom{N-1}{q}}.
\end{equation}

When $p_{s,d-s}$ is of maximal rank for all $s$, Koszul flattenings can only give a better Waring border rank lower bound than the partial derivative maps when $d=2k+1$ is odd and $s=k$. For example, if $d=2k$ is even, then the rank of the Koszul flattening is bounded above by $\dim \left( \Lambda^q \bbC^N  \otimes S^{k-1}\bbC^N \right)$ so the   Waring border rank lower bound from Koszul flattenings is bounded above by $\binom{N+k-2}{k-1}\frac N{N-q}$ whereas from flattenings alone, one already gets the larger lower bound of $\binom{N+k-1}{k}$. Similarly, if $d$ is odd but $N \neq 2q+1$, the bound from Koszul flattenings is lower than the one from standard flattenings. 

Suppose $d=2k+1$ and  $p_{k,k+1}$ is of maximal rank. Viewed na\"\i vely, one might think Koszul flattenings could prove Waring border rank lower bounds of up to
\[
\frac{\dim ( \Lambda^q \BC^N \otimes S^k\BC^{N*})}{\binom{N-1}{q}}=\binom{N+k-1}k \frac N{N-q}
\]
but this is not possible because the exterior derivative map is a $GL_N$-module map  that is not surjective unless $q = N,N-1$. Indeed, we have
the decompositions
\begin{align*} 
\Lambda^q \BC^N \otimes S^{k+1} \BC^N &= S_{k+1,1^q}\BC^N\oplus S_{k+2,1^{q-1}}\BC^N \\
\Lambda^{q+1} \BC^N\otimes S^{k } \BC^N &= S_{k ,1^{q+1}}\BC^N\oplus S_{k+1,1^{q }}\BC^N
\end{align*}
so, by Schur's Lemma, $S_{k+2,1^{q-1}}\BC^N =  \ker (\delta)$ and   $\tim (\delta) = S_{k+1,1^{q }}\BC^N$.
 
The following result gives an \emph{a priori} upper bound for the rank of the Koszul flattening, by determining a lower bound for the dimension of $\ker(p_{k,k+1}^{\wedge q})$.

\begin{proposition}\label{kozapriori} Let $p\in S^{2k+1}\BC^N$. Then for every $q$

\begin{equation}\label{eqn: a priori estimate for Koszul}
\rank(p^{\wedge q}_{k,k+1})\leq \sum_{j=0}^k (-1)^j\dim ( \La{q-j}\BC^N \otimes S^{k-j}\BC^{N*}).
\end{equation}
\end{proposition}
\begin{proof}
We will prove that 
\[
\dim \tim (p^{\wedge(q-1)}_{e-1,d-e+1}) \leq \dim \ker (p^{\wedge q}_{e,d-e}) 
\]
and conclude that the estimate holds for every $q$ via an induction argument. Indeed, $\tim (p^{\wedge (q-1)}_{e-1,d-e+1}) \subseteq \tim (\Id_{\Lambda^{q} \bbC^N} \otimes p_{e,d-e})$ because the exterior derivative map takes derivatives on the factor $S^{d-e+1} \bbC^N$. Moreover, since $\delta^2 = 0$, we have $\tim (p^{\wedge q-1}_{e-1,d-e+1}) \subseteq \ker(\delta_{q,d-e})$. Therefore, passing to dimensions
\[
\dim (\tim (p^{\wedge q-1}_{e-1,d-e+1})) \leq \dim ( \ker(\delta_{q,d-e} \vert_{\tim (\Id_{\Lambda^{q} \bbC^N} \otimes p_{e,d-e})}) \leq \dim \ker (p^{\wedge q}_{e,d-e}).
\]
Now,
\begin{align*} 
\rank(p^{\wedge q}_{k,k+1}) &= \dim  (\La{q}\BC^N \otimes S^{k}\BC^{N*} ) - \dim (\ker (p^{\wedge q}_{k,k+1}))   \\ 
  & \leq  \dim (\La{q}\BC^N \otimes S^{k}\BC^{N*}) - \rank (p^{q-1}_{k-1,k+2}) . 
\end{align*}
and we conclude by induction.
\end{proof}

\begin{remark}\label{rmk: skewsymmetric koszul}
This is still not the end of the story:  when $N=2q+1$ with $q$ odd, then the linear map $p^{\wedge q}_{k,k+1}: \La q \BC^N \otimes  S^k\BC^{N*} \to \La{q+1}\BC^N  \otimes S^{k}\BC^N \simeq  \La{q}\BC^{N*} \otimes S^k\BC^N$ was observed in \cite{MR3081636}, at least in certain cases, to be skew-symmetric. In particular, in this case,  if the bound in \eqref{eqn: a priori estimate for Koszul} is odd, it cannot be attained. 
\end{remark}

\begin{remark}
Since the border rank bound is obtained by dividing $\rank(p^{\wedge q}_{k,k+1})$ by $\binom{N-1}q$, the best potential lower bound is obtained when $N=2q+1$ and there the limit of the method is   twice the bounds obtained via flattenings minus lower order terms. This improvement is irrelevant for complexity. It is known more generally that the improvement in best possible lower bounds beyond the best possible bounds of partial derivatives are limited for any determinantal method. This was observed independently by Efremenko, Garg, Oliveira and Wigderson \cite{EfrGarOliWig:BarriersRankMethods}, and Ga{\l}\k{a}zka \cite{2016arXiv160508005G} for completely different reasons.
\end{remark}

We now show Koszul flattenings can indeed give border rank lower bounds beyond the best lower bound attainable via the method of partial derivatives. 

\begin{proposition}\label{fknkosz}  For all $n > 2$ and all $q < \frac{n}{2}$, the Koszul flattening $(\tilde f_{n,k})_{k,k+1}^{\wedge q}$ has rank at least
\[
\binom{n-1}{q} \Bigl(\binom{n+k-1}{k} +  q-1   \Bigr).
\]
In particular, $\ur_S(\tilde f_{n,k})\geq \binom{n+k-1}{k} +q-1$, which is greater than the lower bound obtainable by flattenings.
\end{proposition} 
\begin{proof}
For fixed $n$, recall the unique $\frakS_n$-invariant $g := g_n \in S^2 [n-1,1]$ from the proof of Theorem \ref{thm: flattenings of ellqk hit everything}. Let $[n]$ be the span of $\ell:= \ell_n$.  For every $s$, write $p_s := (\ell^*)^{s-1} \cdot (\ell q^{s-1})$,  which from Eqn. \eqref{eqn: summation in thm ell q} is a polynomial of degree $s$ with non-zero projection onto $[n]^s$. From Eqn. \eqref{eqn: decomposition of image of ellqk}, the image of the $k$-th flattening map of $\tilde{f}_{n,k}$ is $\tim (\tilde{f}_{n,k})_{k,k+1} = \bigoplus_{s=0}^k p_{s+1}S^{k-s}[n-1,1] \subseteq S^{k+1} \bbC^n$.

We will give a lower bound for the dimension of the image of $ \Lambda^q \bbC^n \otimes \tim (\tilde{f}_{n,k})_{k,k+1}$ under the exterior
derivative map. We have $\Lambda^q \bbC^n = \Lambda^{q} [n-1,1] \oplus \left( [n] \wedge \Lambda^{q-1} [n-1,1] \right)$ as a $\frakS_n$-module.

Consider the image of $( [n]\wedge \La{q-1}[n-1,1]) \otimes (p_{s+1}S^{k-s}[n-1,1])$ under the exterior derivative. The $\frakS_n$-equivariant projection of this space onto $( [n]\wedge \La{q-1}[n-1,1]) \otimes ([n]^{s+1} \otimes S^{k-s}[n-1,1])$ commutes with the exterior derivative, which is $GL_n$-equivariant and therefore $\frakS_n$-equivariant. The image of $( [n]\wedge \La{q-1}[n-1,1]) \otimes ([n]^{s+1} \otimes S^{k-s}[n-1,1])$ under the exterior derivative is (when $s\leq k-1$ and after reordering the factors) the subspace $S_{(k-s,1^{q-1})}[n-1,1]\otimes [n]\otimes [n]^{s+1} \subseteq [n]\wedge \La q[n-1,1]\otimes [n]^{s+1}S^{k-s-1}[n-1,1]$.

Now consider the image of $\La{q }[n-1,1] \otimes p_{s+1}S^{k-s}[n-1,1]$. Again, consider its projection onto $\Lambda^q [n-1,1] \otimes ([n]^{s+1} \otimes S^{k-s}[n-1,1])$. By applying the exterior derivative map, we obtain a subspace of $\left(([n] \wedge \Lambda^q [n-1,1]) \otimes ([n]^{s} \otimes S^{k-s} [n-1,1])\right)\oplus  \left(\La{q+1}[n-1,1]\otimes ([n]^{s+1} \otimes S^{k-s-1}[n-1,1] )\right)$; consider its projection to the second summand $  \La {q+1}[n-1,1]\otimes [n]^{s+1}S^{k-s-1}[n-1,1]$ when $s\leq k-1$. For the same reason as above,  the image of this projection is $S_{(k-s,1^{q })}[n-1,1]\otimes [n]^{s+1}$ up to reordering the factors.

Note that $S_{(k-t+1,1^{q-1})}[n-1,1]\oplus S_{(k-t,1^q)}[n-1,1]=S^{k-t}[n-1,1]\otimes \La{q}[n-1,1]$ as a $GL([n-1,1])$-module. Consider the summands for $s$ ranging from $0$ to $k-2 $ in the first case and $1$ to $k -1$ in the second, we obtain components $S^{k-t}[n-1,1]\otimes \La{q}[n-1,1]$ for $t$ from $1$ to $k-1$. We obtain a subspace in the image of the Koszul flattening that is isomorphic as a $GL([n-1,1])$-module to 
\[
\Bigl( \ \bigoplus_{t=1}^{k-1}S^{k-t}[n-1,1] \Bigr)\otimes \La{q}[n-1,1].
\]
The first factor of the space above has the same dimension as $S^{k-1}\BC^n$ minus $\dim (S^0[n-1,1])=1$.

So far we have a contribution to the rank of
\[
\binom{n-1}{q} \Bigl( \binom{n+k-2}{k-1} - 1 \Bigr).
\]

Next consider the $s=0$ contribution that one obtains by applying the exterior derivative to the component with the factor $\Lambda^q [n-1,1]$. The exterior derivative map is $\La q [n-1,1]  \otimes ([n] \otimes S^k[n-1,1]) \to \La q [n-1,1]\wedge [n]  \otimes S^k[n-1,1] \oplus  \Lambda^q [n-1,1] \otimes ([n]\otimes  S^{k-1}[n-1,1])$. This projects isomorphically onto the first term in the target and $(\La q [n-1,1]\wedge [n]) \otimes S^k[n-1,1]$ does not intersect the image of any other term that we considered so far, so we obtain an additional contribution to the image of dimension $\binom{n-1+k-1}{k}\binom{n-1}q$. We now have a contribution of
\[
\binom{n-1}q \Bigl(\binom{n+k-2}{k-1} - 1+ \binom{n+k-2}{k} \Bigr)=\binom{n-1}q \Bigl(\binom{n+k-1}{k } - 1\Bigr)  
\]
to the rank. 

 Finally consider  the $s=k$ and $s=k-1$ terms in the term with a factor $[n] \wedge \Lambda^{q-1}[n-1,1]$: the sources are respectively $[n]\wedge \La {q-1} [n-1,1] \otimes \langle p_{k+1} \rangle$ and $([n] \wedge \Lambda^{q-1} [n-1,1]) \otimes p_{k }[n-1,1]$. Consider the projections respectively to $([n]\wedge \La {q-1} [n-1,1] )\otimes \ell^{k-1}S^2[n-1,1]$ (the second factor is $\ell^{k-1}g$) and $( [n]\wedge \La {q-1} [n-1,1]) \otimes \ell^{k-2}S^3[n-1,1]$ (the second factor is $\ell^{k-2}g[n-1,1]$). Now, applying the exterior derivative, these spaces map injectively to $([n] \wedge \Lambda^{q}[n-1,1]) \otimes [n]^{k-1} [n-1,1]$ and $([n] \wedge \Lambda^{q}[n-1,1]) \otimes \ell^{k-2}S^2 [n-1,1]$ respectively. These targets do not appear in other terms analyzed above, so we pick up $\binom{n-1}{q-1}=\frac qn \binom{n-1}{q}$ and $(n-1)\binom{n-1}{q-1}=\frac{q(n-1)}n\binom{n-1}q$ additional contributions to the rank.

Collecting all the contributions together, we conclude.
\end{proof}

 \begin{remark}
A more careful analysis of the $q=1$ case shows it also improves the flattening lower bound.
\end{remark}

\begin{remark}\label{koscomputer}
Computer experiments performed in Macaulay2 \cite{M2}, using a variant of the code of \cite{OedOtt:EgnvectorsTensorsAlgorithmsWaringDecomp}, indicate that the situation may be significantly better. 

For $ n = 3 \vvirg 6$, $k = 1 \vvirg 6$, let $p \in S^{2k+1}\bbC^n$ be generic. The Koszul flattening $p_{k,k+1}^{\wedge 1}$ has rank equal to the bound in \eqref{eqn: a priori estimate for Koszul} except if $n=3$ and $k$ is even (in accordance with Remark \ref{rmk: skewsymmetric koszul}).

For  $n = 3 \vvirg 6$ and $k= 1 \vvirg 6$, let $p = h_{n,2k+1}$ or $p = \tilde f_{n,k}$. The Koszul flattenings of $p_{k,k+1}^{\wedge 1}$ has rank equal to the bound in \eqref{eqn: a priori estimate for Koszul} if $k$ is odd and one less than the bound in \eqref{eqn: a priori estimate for Koszul} if $k$ is even.
\end{remark}

Notice that in the cases $n=4,5,6$ with $k$ even, the Koszul flattenings give a border rank lower bound for $h_{n,2k+1}$ and $\tilde f_{n,k}$ that is one less than the bound for a generic polynomial.

\begin{question} What are the ranks of the Koszul flattenings for $\tilde f_{n,k}$ and $h_{n,2k+1}$ in general?
\end{question}

\bibliographystyle{amsalpha}
\bibliography{flatbib}

\end{document}